\documentclass[final]{article}
\usepackage{graphicx}
\usepackage{verbatim}
\usepackage[latin1]{inputenc}
\usepackage[text={420pt,650pt},centering]{geometry}
\usepackage{color}
\usepackage{amsmath}
\usepackage{amssymb}
\usepackage{amsthm}
\newtheorem{thm}{Theorem}[section]

\newtheorem{prop}[thm]{Proposition}
\theoremstyle{definition}
\newtheorem{defn}[thm]{Definition}
\theoremstyle{remark}

\numberwithin{equation}{section}

\newcommand{\Gram}{\mathcal{G}}
\newcommand{\NtSet}{\mathcal{N}}
\newcommand{\Nt}{N}
\newcommand{\T}{t}
\newcommand{\Voc}{\Sigma}
\newcommand{\Axiom}{\mathcal{S}}
\newcommand{\ProdRules}{\mathcal{P}}
\newcommand{\Lang}[1]{\mathcal{L}(#1)}
\newcommand{\Imm}{\mathcal{L}^{\lhd}}
\newcommand{\LangImm}[1]{\mathcal{L}^{\lhd}(#1)}
\newcommand{\BigO}[1]{\mathcal{O}(#1)}
\newcommand{\Def}[1]{{\bf #1}}
\newcommand{\UnG}[2]{#1\;|\;#2}
\newcommand{\PrG}[2]{#1\;.\;#2}
\newcommand{\Production}{\rightarrow}
\newcommand{\Derive}{\Rightarrow}
\newcommand{\DerPol}{\phi}

\newcommand{\AlgW}{\mathcal{A}}
\newcommand{\VectSizes}{\mathbf{n}}
\newcommand{\Remark}[1]{{\bf Remark #1:}}
\newcommand{\Prob}[1]{\mathbb{P}(#1)}
\newcommand{\Expect}[1]{\mathbb{E}(#1)}
\newcommand{\Forb}{\mathcal{F}}
\newcommand{\Pond}{\pi}
\newcommand{\Eqdef}{:=}
\newcommand{\GramW}{\mathcal{G}_\Pond}
\newcommand{\PrefixTree}{\mathcal{T}}

\begin{document}

\title{Non-redundant random generation from weighted context-free languages}%
\author{Yann Ponty\\Department of biology\\Boston College\\Chestnut Hill, MA 02467, USA\\ponty@bc.edu}%


\maketitle
\begin{abstract}
	We address the non-redundant random generation of $k$ words of length $n$ from a context-free language.
	Additionally, we want to avoid a predefined set of words. We study the limits of a rejection-based
	approach, whose time complexity is shown to  grow exponentially in $k$ in some cases. We propose an
	alternative recursive algorithm, whose careful implementation allows for a non-redundant generation of $k$
	words of size $n$ in $\mathcal{O}(kn\log{n})$ arithmetic operations after the precomputation of
	$\Theta(n)$ numbers. The overall complexity is therefore
	dominated by the generation of $k$ words, and the non-redundancy comes at a negligible cost.
	\end{abstract}
\thispagestyle{empty} \pagestyle{empty}

\section{Introduction}
	The random generation of combinatorial objects has many direct applications in areas
	ranging from software engineering~\cite{Den06} to bioinformatics~\cite{PoTeDe06}.
	It can help formulate conjectures on the average-case complexity of algorithms,
	raises new fundamental mathematical questions, and directly benefits from new
	discoveries of its fundamental objects. These include, but are not limited to, generating
	functionology, arbitrary precision arithmetics and bijective combinatorics. Following the
	so-called \emph{recursive} framework introduced by Wilf~\cite{wilf77},
	very elegant and general algorithms for the uniform random generation have been designed~\cite{flajoletcalculus} and implemented.
	Many optimizations of this approach been developed, using specificities of certain classes of
	combinatorial structures~\cite{Gol95}, or floating-point arithmetics~\cite{DeZi99}. More recently
	a probabilistic approach to this problem, the so-called Boltzmann generation~\cite{fullboltz},
	has drawn much attention both because its very low memory complexity and its
	underlying theoretical beauty.

	For many applications, it is necessary to drift away from the \emph{uniform} models.
	A striking example lies in the most recent paradigm for the \emph{in silico} analysis of the RNA molecule's folding.
	Instead of trying to predict a structure of minimal free-energy, current approaches tend to focus on the \emph{ensemble
	properties} of achievable conformations, assuming a Boltzmann probability distribution~\cite{DiLa03}.
	Random generation is then performed, and complex structural features are evaluated in a
	statistical manner. In order to capture such features, a general non-uniform scheme was proposed by Denise
	\emph{et al}~\cite{deniserandom}, based on the concept of \emph{weighted context-free grammars}. Recursive random
	generation algorithms were derived, with time and space complexities similar to that of the uniform ones~\cite{flajoletcalculus}.

	In the weighted probability distribution, the probability ratio between the most and least frequent words
	sometimes grows exponentially on the size of the generated objects. Therefore it is a natural question to address
	the \Def{non-redundant random generation} of combinatorial objects, that is the generation of a
	set of \Def{distinct} objects. By contrast to the general case, this aspect of random generation
	has, to our best knowledge, only  been addressed through the introduction of the {\tt PowerSet} construct by
	Zimmermann~\cite{ZimPowerset95}. An algorithm in $\Theta(n^2)$ arithmetic operations, or a
	practical $\Theta(n^4)$ complexity in this case, was derived for recursive decomposable structures.
	The absence of redundancy in the set of generated structures was achieved respectively through \emph{rejection}
	or an \emph{unranking} algorithms. While the former is discussed later on in the document, the latter cannot
	be transposed directly to the case of weighted languages, since the assumption that different integral ranks
	correspond to different objects does not hold. Furthermore, the algorithm cannot \emph{immediately} account for
	an initial set of forbidden words, short of computing an intersection grammar, a computationally intensive process
	that would further induce large time and memory constants.

	In this paper, we address the non-redundant generation of words from a context-free language, generated while avoiding
	a pre-defined inclusive set $\Forb$. First, we define some concepts and notations, which allows us
	to rephrase the random generation process as a \emph{step-by-step} walk. Then, we investigate the
	efficiency of a rejection-based approach to our problem. We show that, although well-suited for the uniform
	distribution, the rejection approach can yield high average-case complexities for large sets
	of forbidden words. In the weighted case, we show that the complexity of the rejection
	approach can grow exponentially on the number of desired sequences. Finally, we introduce a new algorithm,
	based on a recursive approach, which generates $k$ sequences of length $n$ while avoiding a set $\Forb$
	at the cost of $\mathcal{O}(kn\log(n))$ arithmetic operations after a precomputation in $\Theta((n+|\Forb|)n)$
	arithmetic operations.

\section{Concepts and notations}
	\subsection{Context-free grammars}
	We remind some formal definitions on context-free grammars and Chomsky Normal Form (CNF).
	A \Def{context-free grammar} is a 4-tuple $\Gram=(\Voc,\NtSet,\ProdRules,\Axiom)$ where
	\begin{itemize}
		\item $\Voc$ is the alphabet, i.e. a finite set of terminal symbols.
		\item $\NtSet$ is a finite set of non-terminal symbols.
		\item $\ProdRules$ is the finite set of production rules, each of the form $\Nt\Production X$,
		for $\Nt\in\NtSet$ any non-terminal and $X\in \{\Voc\cup\NtSet\}^*$.
		\item $\Axiom$ is the \Def{axiom} of the grammar, i. e. the initial non-terminal.
	\end{itemize}
	A grammar $\Gram$ is then said to be in \Def{Chomsky Normal Form} (CNF) iff the rules
	associated to each non-terminal $\Nt\in\NtSet$ are either:
	\begin{itemize}
		\item Product case: $\Nt \Production \PrG{\Nt'}{\Nt''}$
		\item Union case: $\Nt \Production \UnG{\Nt'}{\Nt''}$
		\item Terminal case: $\Nt \Production \T$
	\end{itemize}
	for $\Nt,\Nt',\Nt''\in\NtSet$ non-terminal symbols and $\T\in\Voc$ a terminal symbol.
	In addition, the axiom $\Axiom\in\NtSet$ is allowed to derive the empty
	word $\varepsilon$ only if $\Axiom$ does not appear in the right-hand side of any
	production.

	Let $\Lang{\Nt}$ be the language associated to $\Nt\in\NtSet$, that is the set of words on terminal symbols,
	accessible through a sequence of derivations starting from $\Nt$.
	Then the language $\Lang{\Gram}$ generated by a grammar $\Gram=(\Voc,\NtSet,\ProdRules,\Axiom)$
	is defined to be the language $\Lang{\Axiom}$ associated with its axiom $\Axiom$.
	It is a classic result that any context-free grammar $\Gram$ can be transformed into a grammar $\Gram'$ in Chomsky Normal Form (CNF)
	such that $\Lang{\Gram} = \Lang{\Gram'}$. Therefore, we will define our
	algorithm on CNF grammars, but will nevertheless illustrate its behavior on a compact, non-CNF, grammar for Motzkin words.
	Indeed the normalization process introduces numerous non-terminals even for the most trivial grammars.

	\subsection{Fixed-length languages descriptions: Immature words}
	We call \Def{mature} word a sequence of terminal symbols. More generally, we will call \Def{immature}
	a word that contains both terminal and non-terminal symbols, thus potentially requiring some further
	processing before becoming a mature word. We will denote $\LangImm{\Nt}$ the set of immature words
	accessible from a non-terminal symbol $\Nt$ and extend this notion to $\LangImm{\Gram}$
	the immature words accessible from the axiom of a grammar $\Gram$. It is noteworthy that
	$\Lang{\Gram}\subset\LangImm{\Gram}$.

	We can then \Def{attach required lengths} to the symbols of an immature word.
	For instance, in the grammar for Motzkin words from figure~\ref{fig:MotzkinWalks},
	$c\;a\;S_4\;b\;S_0$ will be a specialization of the immature word $c\;a\;S\;b\;S$,
	where the words generated from the first (resp. second) instance of the non-terminal $S$
	are required to have length $4$ (resp. $0$). Formally, this amounts to taking into consideration
	couples of the form  $(\omega,\VectSizes)$ where $\omega$ is an immature word, and
	$\VectSizes\in\mathbb{N}^{|\omega|}$ is a vector of sizes for words generated from the different
	symbols of $\omega$. We naturally extend the notion of language associated with an immature word
	to these couples in the following way:
	$$ \Lang{\omega,\VectSizes} = \prod_{i\ge 1}^{|\omega|}{\Lang{\omega_i,\VectSizes_i}} $$
	The length vector ${\bf n}$ associated with an immature word $\omega$ may
	be omitted in the rest of the document for the sake of simplicity.

		\begin{figure}[t]
				\begin{center}
				\fbox{
					 \begin{tabular}{ccc}
					 \includegraphics[height=4cm]{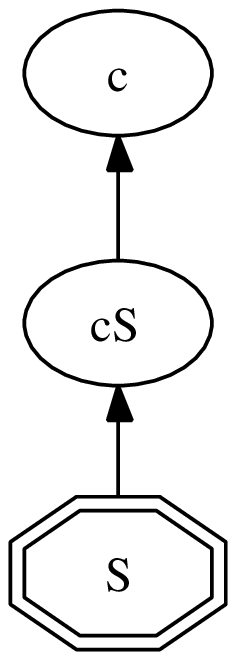} & \includegraphics[height=4cm]{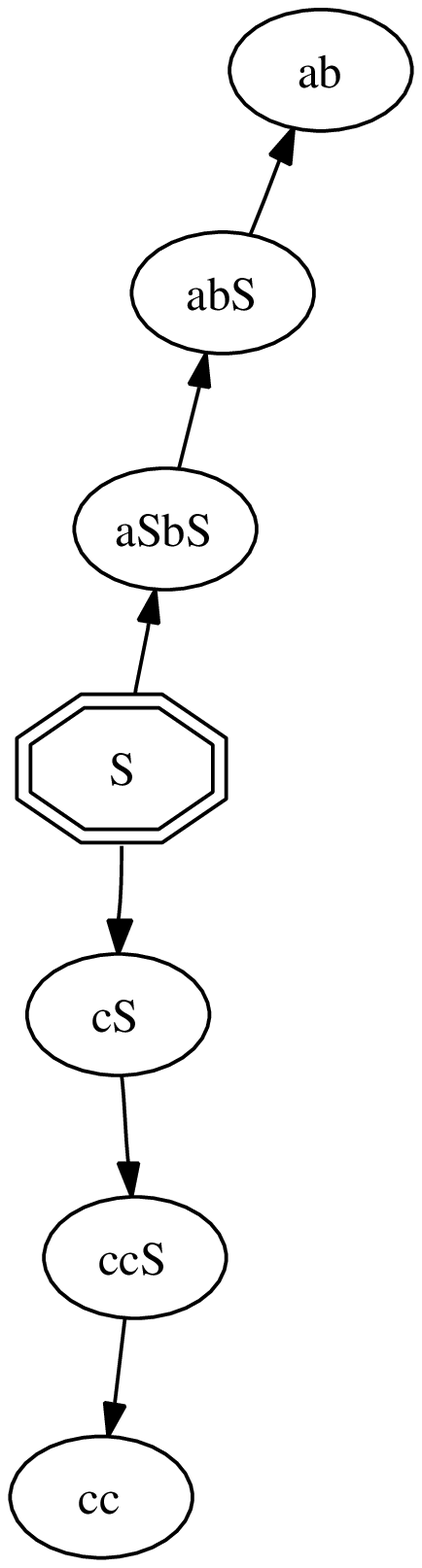} & \includegraphics[height=4cm]{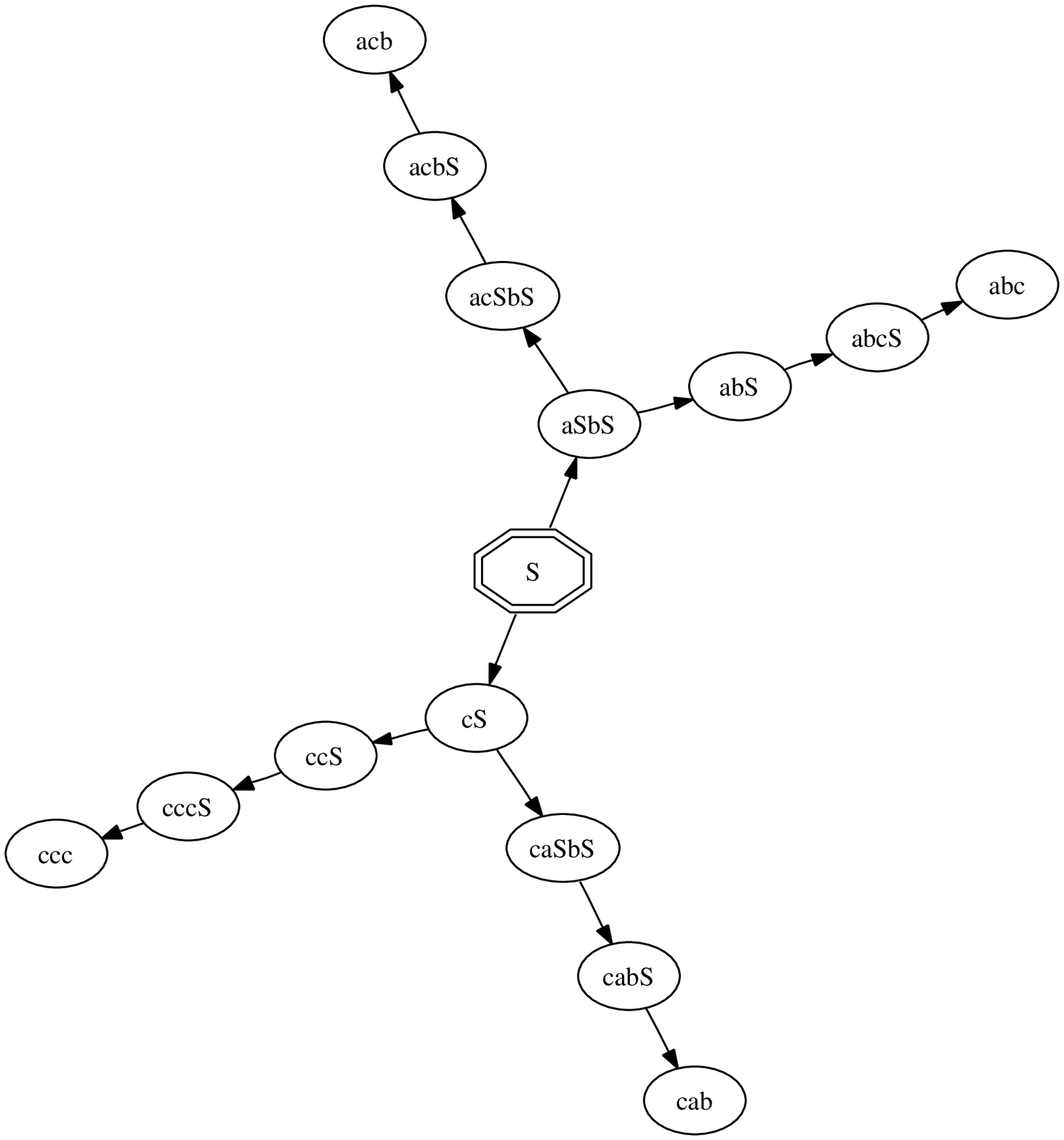}\\
					 $n=1$ & $n=2$ & $n=3$\\
					 \includegraphics[height=3cm]{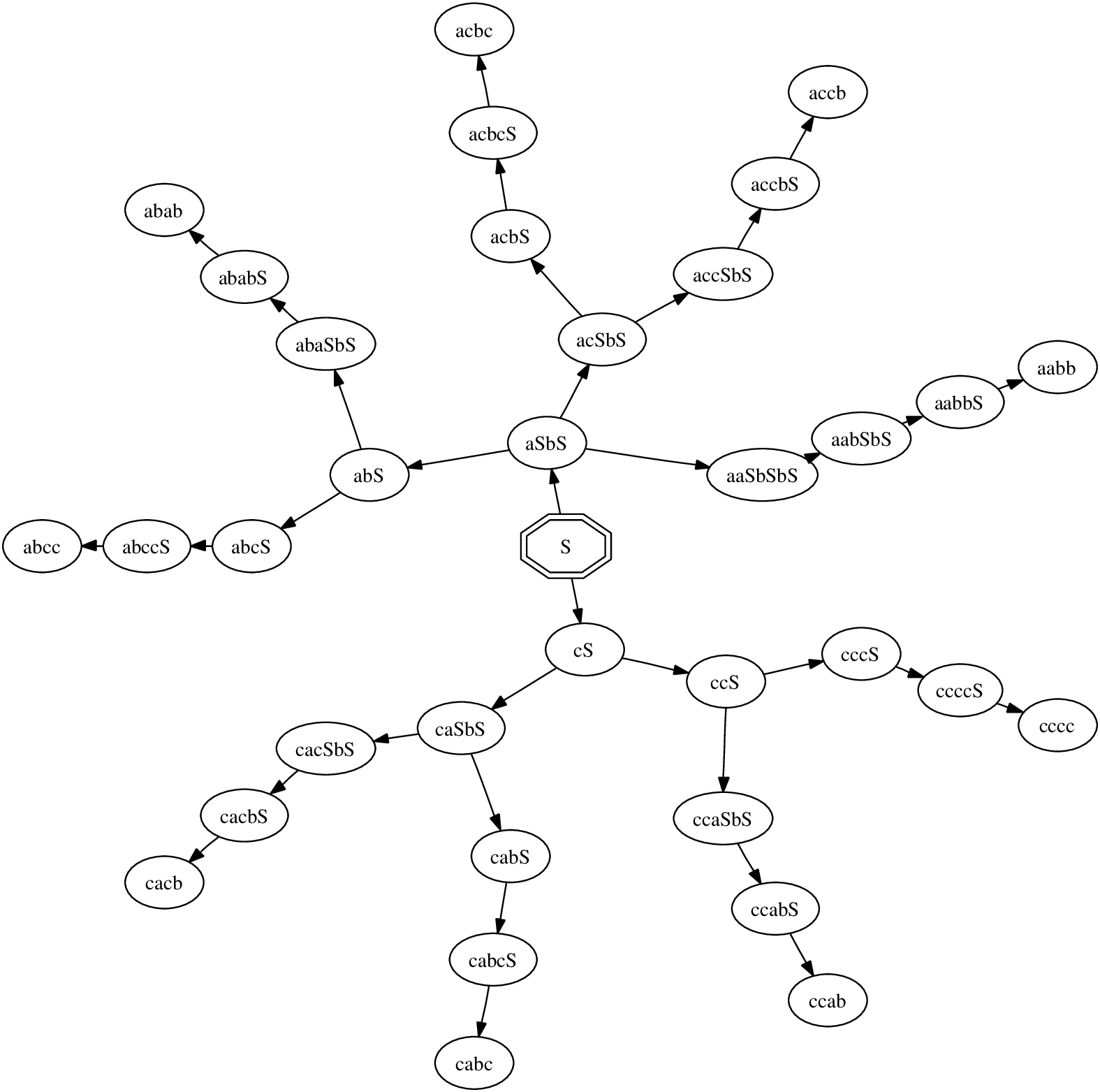} & \includegraphics[height=3cm]{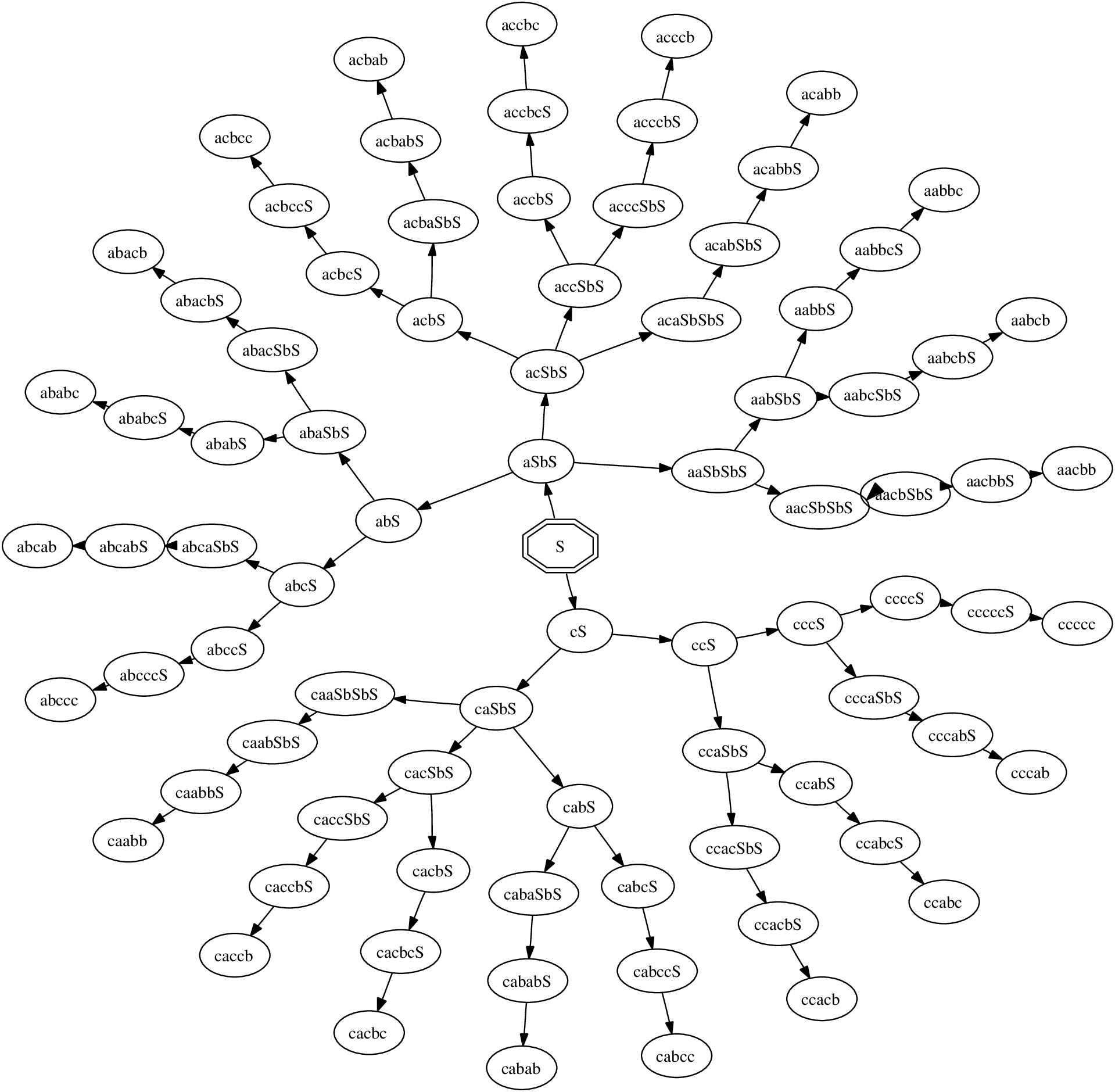} & \includegraphics[height=3cm]{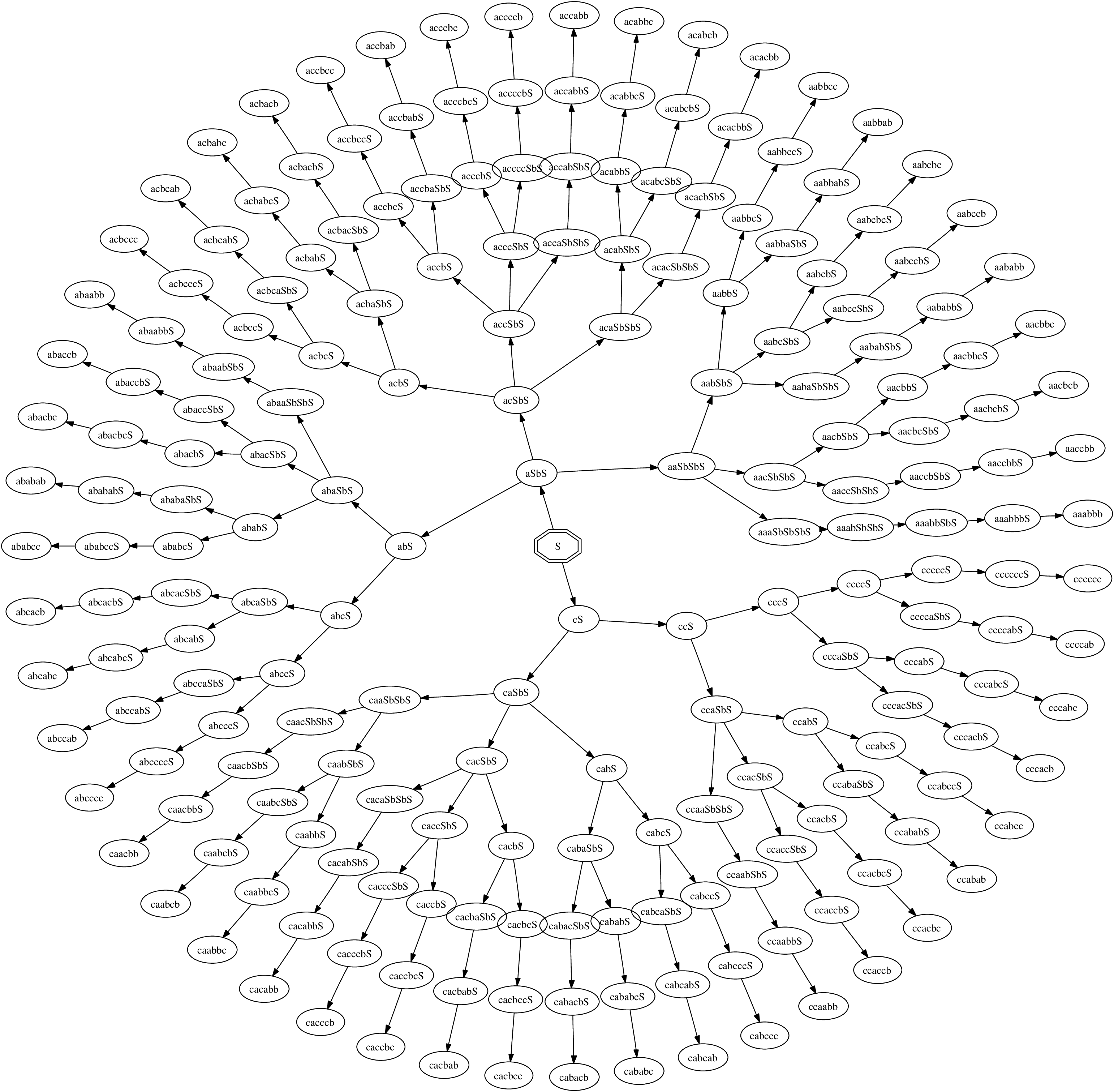}\\
					 $n=4$ & $n=5$ & $n=6$
					 \end{tabular}}
				\end{center}
				\caption{Trees of all walks associated with Motzkin words of size $n\in[1,6]$ generated by the grammar $S\Production \UnG{\UnG{a\;S\;b\;S}{c\; S}}{\varepsilon}$
				 under a \emph{leftmost first} derivation policy.}
				\label{fig:MotzkinWalks}
	\end{figure}
	\subsection{Random generation as a random walk in language space}
	An \Def{atomic derivation}, starting from a word $\omega=\PrG{\PrG{\omega'}{\Nt}}{\omega''}\in\{\Voc\cup\NtSet\}^*$,
	is the application of a production $\Nt\Production X$ to $\omega$ that replaces $\Nt$ by the right-hand
	side $X$ of the production, which yields $\omega\Derive \omega'.X.\omega''$.
	Let us call \Def{derivation policy} a deterministic strategy that points, in an immature word, the
	first non-terminal to be rewritten through an atomic derivation. Formally, in the context of a grammar $\Gram$,
	a derivation policy is a function $\DerPol:\Lang{\Gram}\cup\LangImm{\Gram}\to\mathbb{N}\cup\{\emptyset\}$ such
	that
	$$
		\begin{array}{rrcl}
			\DerPol : & \omega\in\Lang{\Gram}& \to & \emptyset\\
							 & \omega'\in\LangImm{\Gram}& \to & i\in [1,|\omega'|]
		\end{array}
	$$

	A sequence of atomic derivations is then said to be \Def{consistent with a given derivation policy} if
	the non-terminal rewritten at each step is the one pointed by the policy.
	A side effect of this somewhat verbose notion is that it provides a convenient
	framework for defining the \Def{unambiguity} of a grammar without any reference to parse trees.
	\begin{defn}
		Let $\Gram=(\Voc,\NtSet,\ProdRules,\Axiom)$ be a context-free grammar and $\DerPol$ a derivation policy acting on $\Gram$.
		The grammar $\Gram$ is said unambiguous if and only, for each $\omega\in\Lang{\Gram}$, there exists
		only one sequence of atomic derivations that is consistent with  $\DerPol$ and produces $\omega$ from $\Axiom$.
	\end{defn}

		The derivation leading to a mature word $\omega\in\Lang{\Gram}$ in a grammar $\Gram$ can then be associated
	in a one-to-one fashion with a walk in the space of languages associated with immature words, or \Def{parse walk},
	taking steps consistent with a given derivation policy $\phi$. More
	precisely, such walks starts from the axiom $\Axiom$ of the grammar.
	From a given immature word $X\in\LangImm{\Gram}$, the derivation policy $\phi$ points at a position
	$k\Eqdef\phi(X)$, where a non-terminal $X_k$ can be found. The parse walk can then be
	prolonged using one of the derivations acting on $X_k$ (See Figures~\ref{fig:MotzkinWalks}
	and~\ref{fig:FullExample}).

	\subsection{Weighted context-free grammars}
	\begin{defn}[Weighted Grammar~\cite{deniserandom}]
		A weighted grammar $\GramW$ is a 5-tuple $\GramW=(\Pond,\Voc,\NtSet,\ProdRules,\Axiom)$ where
		$\Voc$, $\NtSet$, $\ProdRules$ and $\Axiom$ are the members of a context-free grammar, and
		$\Pond:\Voc\to\mathbb{R}$ is a weighting function, that associates a real-valued weight $\Pond_\T$ to each
		terminal symbol $\T$.
	\end{defn}
	This notion of weight naturally extends to a mature word $w$ in a multiplicative fashion,
	i.e. such that $\Pond(w)=\prod_{i=1}^{|w|}\Pond_{w_i}$. From this, we can define a
	$\Pond$-\Def{weighted probability distribution} over a set of words  $\mathcal{L}$, such that
	the probability associated with $\omega\in\mathcal{L}$ is
		$$ \Prob{w\;|\;\Pond} = \frac{\displaystyle{\Pond(w)}}{\displaystyle{\sum_{w'\in\mathcal{L}}\Pond(w')}} $$
	In the rest of the document, we may use $\Pond(\omega)$ instead of $\Pond(\Lang{\omega,{\bf n}})$ to denote the total weight
	of all words derivable from an immature word $\omega$ with associated lengths ${\bf n}$.

\section{Efficiency of a rejection approach for the non-redundant generation}\label{sec:rejet}
	We address the uniform generation of a non-redundant set of words from a language
	$\mathcal{L}$ with \Def{forbidden words} $\Forb\subset\mathcal{L}$. A \Def{rejection
	approach} for this problem consists in drawing words at
	random in an unconstrained way, rejecting those previously sampled until $k$ distinct words
	are generated, as prescribed by Zimmermann~\cite{ZimPowerset95} in the case of recursive specifications.
	For the generation of words from context-free languages, we refer to
	previous works of Flajolet \emph{et al}~\cite{flajoletcalculus,fullboltz}, or
	Denise \emph{et al}~\cite{DeZi99} that achieves  a $\BigO{n^{1+\varepsilon}}$
	complexity through highly non-trivial floating-point arithmetics.
\subsection{The uniform case}
	\begin{thm}
		Let $\mathcal{L}$ be a context-free language, $n\in\mathbb{N}^+$ a positive integer and $\Forb\subset\mathcal{L}_n$ a
		set of forbidden words. Then the \Def{rejection approach} for the non-redundant uniform random generation of $k$
		words of size $n$ from $\mathcal{L}$ has average-case complexity in
		$\mathcal{O}\left(\left(\frac{|\mathcal{L}_n|}{|\mathcal{L}_n|-|\Forb|}\right)n^{1+\varepsilon}k\log{k}\right)$.
	\end{thm}
	\begin{proof}
		In the uniform model when $\Forb=\emptyset$, the number of attempts necessary to the generation of the
		$i$-th word only depends on $i$ and is independent from prior events. Thus the random variable $X_{n,k}$ that
		contains the total number of trials for the generation of $k$ words of size $n$ is such that
		$$ \Expect{X_{n,k}} = \sum_{i=0}^{k-1}\frac{l_n}{l_n-i} = l_n(\mathcal{H}_{l_n}-\mathcal{H}_{l_n-k}) $$
		where $l_{n}\Eqdef|\mathcal{L}_n|$ is the number of words of size $n$ in the language and $\mathcal{H}_i$ the harmonic number of order $i$, as
		pointed out by Flajolet \emph{et al}~\cite{FlaGarThi92}. It follows that $\Expect{X_{n,k}}$ is trivially
		increasing with $k$, while remaining upper bounded by $k\mathcal{H}_{k}\in\Theta(k\log(k))$ when $k=l_n$
		(Coupon collector problem). Since the expected number of rejections due to a non-empty forbidden set $\Forb$ remains the same
		throughout the generation, and does not have any influence over the generated sequences, it can be
		considered independently and contributes to a factor $\frac{|\mathcal{L}_n|}{|\mathcal{L}_n|-|\Forb|}$.
		Finally, each generation takes time $\BigO{n^{1+\varepsilon}}$, independently from both the generated sequence
		and the cardinality of $\mathcal{L}_n$.
	\end{proof}
	The complexity of a rejection approach to this problem is then mainly linear, unless the set
	$\Forb$ \emph{overwhelms} the language generated by the grammar. In this case, the generation can
	become linear in the cardinality of $\mathcal{L}_n$, that is exponential in $n$ for most languages.
	Furthermore, the worst-case time complexity of this approach remains unbounded.

		\subsection{Weighted context-free languages}
	By contrast with the uniform case, the rejection approach to the non-redundant random
	generation for {\bf weighted context-free languages} can yield an {\bf exponential complexity}, even when starting from an
	empty set of forbidden words $\Forb=\emptyset$. Indeed, the weighted
	distribution can specialize into a power-law distribution on the number of occurrences of
	the terminal symbol having highest weight, leading to an exponentially-growing number of
	rejections.

	{\bf\noindent Example: }Consider the following grammar, generating the language $a^*b^*$ of words starting with
	an arbitrary number of $a$, followed by any number of $b$:
	\begin{eqnarray*}
		S & \Production & \PrG{a}{S}\;|\; T\\
		T & \Production & \PrG{b}{T}\;|\; \varepsilon
	\end{eqnarray*}
	We adjoin a weight function $\Pond$ to this grammar, such that $\Pond(b)\Eqdef\alpha>1$ and $\Pond(a)\Eqdef 1$.
	The probability of the word $\omega_m\Eqdef a^{n-m}b^m$ among $\mathcal{S}_n$ is then
	$$\Prob{\omega_m}
		=\frac{\Pond(\omega_m)}{{\sum_{\substack{\omega\in\Lang{S}\\|\omega|=n}} \Pond(\omega)}}
		=\frac{\alpha^m}{\sum_{i=0}^{n}\alpha^i}
		=\frac{\alpha^{m+1}-\alpha^{m}}{\alpha^{n+1}-1}
		< \alpha^{m-n}.$$
	Now consider the set $\mathcal{V}_{n,k}\subset\mathcal{S}_n$ of words having less
	than $n-k$ occurrences of the symbol $b$. The probability of generating a word from
	$\mathcal{V}_{n,k}$ is then
	$$
		\Prob{\mathcal{V}_{n,k}}
		=\sum_{i= 0}^{n-k}\Prob{\omega_{n-k-i}}
		=\frac{\alpha^{n-k+1}-1}{\alpha^{n+1}-1}< \alpha^{-k}
	$$
	The expected number of generations before a sequence from $\mathcal{V}_{n,k}$ is generated
	is then lower-bounded by $\alpha^k$. Since any non-redundant set of $k$ sequences issued from $\mathcal{S}_n$
	must contain at least one sequence from $\mathcal{V}_{n,k}$, then the average-case time complexity
	of the rejection approach is in $\Omega(n\alpha^k)$, that is exponential in $k$ the number of words.
	\\

	One may argue that the previous example is not very typical of the rejection algorithm's
	behavior on a general context-free language, since the grammar is left linear and consequently
	defines a rational language.
	By contrast, it can be shown that, under a natural assumption, no word
	can asymptotically contribute up to a significant proportion of the distribution in simple type grammars.
	\begin{thm}
		Let $\GramW=(\Pond,\Voc,\NtSet,\Axiom,\ProdRules)$ be a weighted grammar of simple type\footnote{A grammar of
		simple type is mainly a grammar whose dependency graph is strongly-connected and whose number of words follow
		an aperiodic progression (See~\cite{FlaFusPiv07} for a more complete definition). Such a grammar can easily be
		found for the avatars of the algebraic class of combinatorial structures (Dyck words, Motzkin paths, trees of
		fixed degree,...), all of which can be interpreted as trees.}. Let $\omega^0_n$ be the word of length $n$
		generated from $\GramW$ with highest probability (i.e. weight) in the $\Pond$-weighted distribution.
		Additionally, let us assume that there exists $\alpha,\kappa\in\mathbb{R}^+$ positive constants such that
		$\Pond(\omega^0_n) \underset{n\to\infty}{\longrightarrow}\kappa\alpha^n$.\\
		Then the probability of $\omega^0$ tends to $0$ when $n\to\infty$:
		$$ \Prob{\omega^0\;|\;\Pond}=\frac{\Pond(\omega^0)}{\Pond(\Lang{\GramW}_n)} \underset{n\to\infty}{\longrightarrow} 0 $$
	\end{thm}
	\begin{proof}
		From the Drmota-Lalley-Woods theorem~\cite{Drmota97,Lalley93,Woods93}, we know that the generating
		function of a simple type grammar has a \emph{square-root type} singularity, and its coefficients
		admits an expansion of the form $\frac{\kappa\beta^{n}}{n\sqrt{n}}(1+\mathcal{O}(1/n))$.
		This property holds in the case of weighted context-free grammars, where the coefficients are now the
		overall weights $W_n\Eqdef \Pond(\Lang{\GramW}_n)$. Since $\omega^0_n$ is contributing to $W_n$, then
		$\Pond(\omega^0_n)\le\Pond(\Lang{\GramW}_n)$ and therefore $\beta>\alpha$.
	\end{proof}
	Lastly it can be shown that, for any fixed length $n$, the set of words $\mathcal{M}$
	having maximal number of occurrences of a given symbol $\T$  have total probability $1$
	when $\Pond(\T)\to\infty$. It follows that sampling more than $|\mathcal{M}|$ words can
	be extremely time-consuming if one of the weights dominates the others by several orders
	of magnitude.

\section{Step-by-step approach to the non-redundant random generation}\label{sec:unifForb}
		\begin{figure*}[t]
				\begin{center}
					 \fbox{\begin{tabular}{cc}
							 \multicolumn{2}{c}{\scalebox{0.36}{\input{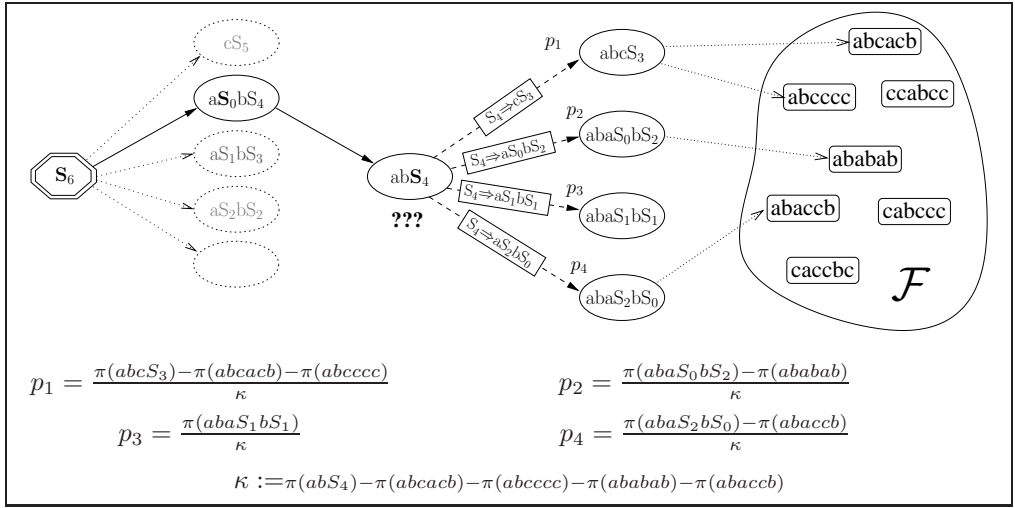}}}\\[0.3cm]
							 $p_1=\frac{\Pond(abcS_3)-\Pond(abcacb)-\Pond(abcccc)}{\kappa}$ & $p_2=\frac{\Pond(abaS_0bS_2)-\Pond(ababab)}{\kappa}$  \\[0.2cm]
							 $p_3=\frac{\Pond(abaS_1bS_1)}{\kappa}$ & $p_4=\frac{\Pond(abaS_2bS_0)-\Pond(abaccb)}{\kappa}$\\[0.2cm]
							 \multicolumn{2}{c}{$\kappa \Eqdef \scriptstyle{\Pond(abS_4)-\Pond(abcacb)-\Pond(abcccc)-\Pond(ababab)-\Pond(abaccb)}$}
						 \end{tabular}}
				\end{center}
				\caption{Snapshot of a \emph{step-by-step} random scenario for a Motzkin word of
				length $6$, generated while avoiding $\Forb$. From $abS_4$, the recursive approach prescribes
				that the probabilities  associated with candidate derivations must be proportional to the overall weights for
				resulting immature words. Additionally, we need to subtract the contributions of $\Forb$ to these immature words,
				which yields probabilities $(p_1,p_2,p_3,p_4)$.}
				\label{fig:FullExample}
	\end{figure*}

	Instead of approaching the random generation from context-free languages by considering non-terminal symbols
	as \Def{independent generators}~\cite{flajoletcalculus,deniserandom}, we consider the random generation
	scenarios as random (parse) walks. This allows to determine to which of the local alternatives
	to subtract the contributions of forbidden (or already generated) words.
	\subsection{The algorithm}
	We propose the following Algorithm $\AlgW$, which from
	\begin{itemize}
		\item a weighted grammar $\GramW=(\Pond,\Voc,\NtSet,\ProdRules,\Axiom)$,
		\item an immature word $(\omega,{\bf n})$,
		\item and a set of forbidden words $\Forb\in\Lang{(\omega,{\bf n})}$,
	\end{itemize}
	draws at random a word $\Lang{\omega,{\bf n}}/\Forb$ with respect to a $\Pond$-weighted
	distribution:
	\begin{enumerate}
		\item \label{alg:AlgoCBegin} If $\omega$ is a mature word \\
						then if $\VectSizes=(1,\ldots,1)$ and $\omega\notin\Forb$ then return $\omega$ else {\tt Error}.
		\item
		Let $\Nt^*_m$ be the non-terminal pointed by $\DerPol$
		in $\omega$, such as $\omega=\omega'.\Nt^*_m.\omega''$
		and $k^\Pond$ be the total weight of words from $\Forb$ generated from $\omega$.
		\item \label{alg:AlgoCRewrite}
		Choose a derivation $\omega\Derive \omega' X\omega''$ with the following probabilities,
		depending on the type of $\Nt^*$:
		\begin{itemize}
			\item $\Nt^* \Production \UnG{\Nt'}{\Nt''}$: Let $k^\Pond_{\Nt'}$ be the sum of weights for words from $\Forb$
			generated from $\omega'.\Nt'_m.\omega''$, then
				$$\Prob{X = \Nt'_m} = \displaystyle{\frac{\Pond(\Lang{\omega'.\Nt'_m.\omega''})-k^{\Pond}_{\Nt'}}{\Pond(\Lang{\omega})-k^{\Pond}}} = 1-\Prob{X =
				\Nt''_m}$$
			\item $\Nt^* \Production \PrG{\Nt'}{\Nt''}$: Let $k^\Pond_{i}$ be the sum of weights for words from $\Forb$
			generated from $\omega'.\Nt'_i.\Nt''_{m-i}.\omega''$.
			$$ \Prob{X = \Nt'_i.\Nt''_{m-i}} =
				 \displaystyle{\frac{\Pond(\Lang{\omega'.\Nt'_i.\Nt''_{m-i}.\omega''})-k^{\Pond}_i}{\Pond(\Lang{\omega})-k^{\Pond}}},\;i\in[1,m-1] $$
			\item $\Nt^* \Production \T$: If $m=1$ then $ \Prob{X = \T}=1 $  else {\tt Error}.
		\end{itemize}
		\item Iterate from step \ref{alg:AlgoCBegin}.
	\end{enumerate}
	\begin{prop} \label{thm:lastTheorem}
		Let $\GramW=(\Pond,\Voc,\NtSet,\Axiom,\ProdRules)$ be a weighted grammar, $\omega\in\{\Voc\cup\NtSet\}^*$ be an immature word
		and $\VectSizes\in\mathbb{N}^{|\omega|}$ be a vector of sizes associated with positions of $\omega$.\\
		Then Algorithm $\AlgW$ draws a mature word at random according to the $\Pond$-weighted distribution
		from $\Lang{\omega,\VectSizes}\backslash\Forb$ or throws an Error if
		$\Lang{\omega,\VectSizes}=\emptyset$.
	\end{prop}
	\begin{proof}
		The previous claim can be proved very quickly by induction on the number $k$ of required executions of line
		\ref{alg:AlgoCBegin} before a mature word is issued:

		{\bf\noindent Base:} The $k=0$ case corresponds to an already mature word $\omega$, for which the associated language is limited
		to $\{\omega\}$. As $\omega$ is generated by line \ref{alg:AlgoCBegin} iff $\omega\notin\Forb$ then the claimed
		result holds in that case.

		{\bf\noindent Inductive step:} Assuming that the theorem holds for $k\ge n$, we investigate the probabilities of emission for
		words that require $k=n+1$ derivations. Let $\Nt^*_m$ be the non-terminal pointed by $\DerPol$, then:
		\begin{itemize}
			\item $\Nt^* \Production \UnG{\Nt'}{\Nt''}$:
				Assume that the derivation $\Nt^*_m\Derive \Nt'_m$ is chosen w.p.
				$\frac{\Pond(\Lang{\omega'.\Nt'_m.\omega''})-k^\Pond_{\Nt'}}{\Pond(\Lang{\omega})-k^\Pond}$.
				Then the induction hypothesis applies and a word $x$ is generated from $\Lang{\omega'.\Nt'_m.\omega''}\backslash\Forb$
				in the $\Pond$-weighted distribution. In this distribution applied to $\Lang{\omega'.\Nt'_m.\omega''}\backslash\Forb$, the probability of $x$ is then
				\begin{eqnarray*}
					\Prob{x\;|\;\omega'.\Nt'_m.\omega''} &=& \frac{\Pond(x)}{\Pond(\Lang{\omega'.\Nt'_m.\omega''}\backslash\Forb)}\\
					&=&\frac{\Pond(x)}{\Pond(\Lang{\omega'.\Nt'_m.\omega''})-\Pond(\Lang{\omega'.\Nt'_m.\omega''}\cap\Forb)}\\
					&=&\frac{\Pond(x)}{\Pond(\Lang{\omega'.\Nt'_m.\omega''})-k^\Pond_{\Nt'}}
				\end{eqnarray*}
				The overall probability of $x$ starting from $\Nt^*_m$ is then
				\begin{eqnarray*}
				 \Prob{x} & = &
				 \frac{\Pond(\Lang{\omega'.\Nt'_m.\omega''})-k^\Pond_{\Nt'}}{\Pond(\Lang{\omega})-k^\Pond}\cdot\frac{\Pond(x)}{\Pond(\Lang{\omega'.\Nt'_m.\omega''})-k^\Pond_{\Nt'}}\\
					& = & \frac{\Pond(x)}{\Pond(\Lang{\omega})-k^\Pond} = \frac{\Pond(x)}{\Pond(\Lang{\omega}\backslash \Forb)}
				\end{eqnarray*}
				This property also holds if $\Nt''_m$ is chosen, after pointing out that
				$$\left(k^\Pond_{\Nt''}=k^\Pond - k^\Pond_{\Nt'}\right)\Rightarrow
				\left(1-\Prob{X = \Nt'_m}=\frac{\Pond(\Lang{\omega'.\Nt''_m.\omega''})-k^\pi_{\Nt'}}{\pi(\Lang{\omega})-k^{\pi}}\right)$$
			\item $\Nt^* \Production \PrG{\Nt'}{\Nt''}$:
				For any $i\in[1,m-1]$, a partition $\Nt^*_m \Derive \PrG{\Nt'_i}{\Nt''_{m-i}}$ is chosen w.p.
				$\frac{\Pond(\Lang{\omega'.\Nt'_i.\Nt''_{m-i}.\omega''})-k^\Pond_{i}}{\Pond(\Lang{\omega})-k^\Pond}$.
				Then the induction hypothesis applies and a word $x$ is generated from $\Lang{\omega'.\Nt'_i.\Nt''_{m-i}\omega''}\backslash\Forb$
				in the $\Pond$-weighted distribution. In this distribution applied to $\Lang{\omega'.\Nt'_i.\Nt''_{m-i}.\omega''}\backslash\Forb$,
				the probability of $x$ is then
				\begin{eqnarray*}
					\Prob{x\;|\;\omega'.\Nt'_i.\Nt''_{m-i}.\omega''} &=& \frac{\Pond(x)}{\Pond(\Lang{\omega'.\Nt'_i.\Nt''_{m-i}.\omega''}\backslash\Forb)}\\
					 &=&\frac{\Pond(x)}{\Pond(\Lang{\omega'.\Nt'_i.\Nt''_{m-i}.\omega''})-\Pond(\Lang{\omega'.\Nt'_i.\Nt''_{m-i}.\omega''}\cap\Forb)}\\
					&=&\frac{\Pond(x)}{\Pond(\Lang{\omega'.\Nt'_i.\Nt''_{m-i}.\omega''})-k^\Pond_{i}}
				\end{eqnarray*}
				The overall probability of $x$ starting from $\Nt^*_m$ is then
				\begin{eqnarray*}
				 \Prob{x} & = &
				 \frac{\Pond(\Lang{\omega'.\Nt'_i.\Nt''_{m-i}.\omega''})-k^\Pond_{i}}{\Pond(\Lang{\omega})-k^\Pond}\cdot\frac{\Pond(x)}{\Pond(\Lang{\omega'.\Nt'_i.\Nt''_{m-i}.\omega''})-k^\Pond_{i}}\\
					& = & \frac{\Pond(x)}{\Pond(\Lang{\omega})-k^\Pond} = \frac{\Pond(x)}{\Pond(\Lang{\omega}\backslash \Forb)}
				\end{eqnarray*}
			\item $\Nt^* \Production \T$: The probability of any word $x$ issued from $\omega$ is that of the word
			issued from $\omega'.\T.\omega''$, that is $\frac{\Pond(x)}{\Pond(\Lang{\omega'.\T.\omega''}\backslash\Forb)}=\frac{\Pond(x)}{\Pond(\Lang{\omega}\backslash\Forb)}$ by the
			induction hypothesis.
		\end{itemize}
	\end{proof}
	This algorithm then performs random generation of $k$ distinct words from a (weighted) context-free language
	by setting the initial immature word to $\Axiom_n$ the axiom of $\GramW$, adding the freshly generated sequence to $\Forb$
	at each step.

\section{Complexities and data structures}
	The algorithm's complexity depends critically on efficient strategies and data structures for:
	\begin{enumerate}
		\item The weights of languages associated with immature words.
		\item The contributions $k_\Pond$ of forbidden (or already generated) words
		associated with each immature word $(\omega,{\bf n})$
		\item The investigation of the different \emph{partitions} $\Nt^*_m \Derive \PrG{\Nt'_i}{\Nt''_{m-i}}$ in the case of \emph{product rules}.\label{alg:Boustrophedon}
		\item Big numbers arithmetics.\label{alg:Floats}
	\end{enumerate}
	\subsection{Weights of immature languages}
		\begin{prop}
		Let $(\omega,{\bf n})$ be an immature word and its associated length vector, whose weight
		$\Pond(\omega,{\bf n})$ is known.\\
		Then a pre-computation involving $\Theta(n)$ arithmetic operations makes it possible to compute in $\mathcal{O}(1)$ arithmetic operations
		the weight of any  word $(\omega',{\bf n}')$ atomically derived from $(\omega,{\bf n})$ .
		\end{prop}
		\begin{proof}
		In order to compute $\Pond(\omega,{\bf n})$, we first compute the total weights of languages associated
		with each non-terminal $\Nt$ for all sizes up to $n=\sum_{n_i\in {\bf n}}n_i$, as done by the
		traditional approach~\cite{deniserandom}. This can be done in $\Theta(n)$ arithmetic operations, thanks to
		the holonomic nature of the generating functions at stake. Indeed, the coefficients of an holonomic
		function obey to a linear recurrence with polynomial coefficients, which can be algorithmically determined
		(Using the {\tt Maple} package {\tt GFun}~\cite{SalZim94}, for instance).

		We can in turn use these values to compute \emph{on the fly} the weights  of immature words of interest. Namely, while
		rewriting an immature word $\omega \Eqdef \alpha.\Nt^*_k.\beta$
		into $\omega'\Eqdef\alpha.X.\beta$ through a derivation $\Nt^*_k\Derive X$, the new weight $\Pond(\omega')$ is given by
		$$ \Pond(\omega') =  \Pond(\alpha.X.\beta) = \frac{\Pond(\alpha)\Pond(X)\Pond(\beta)\Pond(\Nt^*_k)}{\Pond(\Nt^*_k)} = \frac{\Pond(\omega)\Pond(X)}{\Pond(\Nt^*_k)}$$
		where $X$ contains at most two terms (CNF) and therefore a constant number of arithmetic operations is involved.
		\end{proof}
	\subsection{A prefix tree for forbidden words}
		\begin{figure*}[t]
				\begin{center}
					 \fbox{\scalebox{0.2}{\input{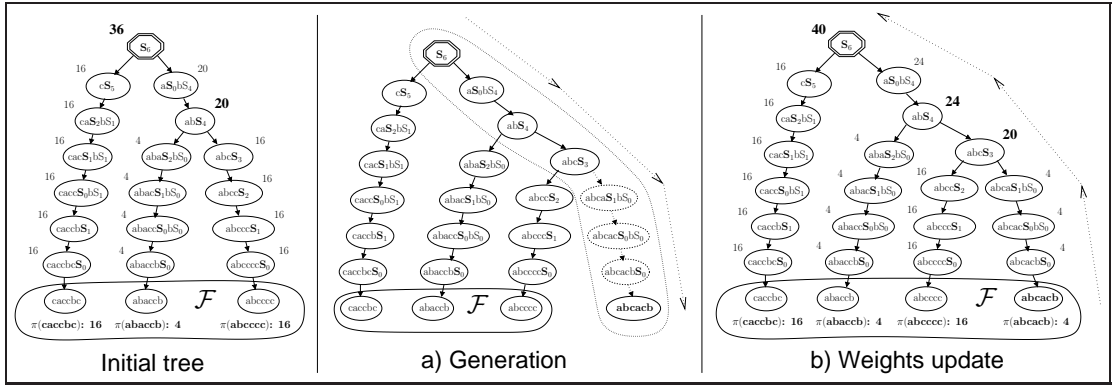}}}
				\end{center}
				\caption{Prefix tree built from the subset $\Forb=\{\mbox{{\bf caccbc}},\mbox{{\bf abaccb}},\mbox{{\bf abcccc}}\}$
				of Motzkin words of size $6$, using weights $\Pond(a)=\Pond(b)=1$ and $\Pond(c)=2$. Generation ({\bf a}) of a new word {\bf abcacb}
				and update ({\bf b}) of the contributions of $\Forb$ to the immature words involved in the generation.}
				\label{fig:PrefixTree}
	\end{figure*}
		\begin{prop}Let $\Forb$ be a set of words characterized by their parse walks\footnote{Starting from an unparsed \emph{raw} set of
		forbidden words $\Forb$ will require some additional parsing, which can be performed in $\Theta(|\Forb|.n^3)$ by a CYK-type algorithm.}.
			Let $\PrefixTree$ be a prefix tree built from the parse walks from $\Forb$ in $\Theta(|\Forb|.n)$ arithmetic operations.
			Then for any immature word $\omega$ reached during a random scenario, the contribution $k_\Pond\Eqdef\Pond(\Lang{\omega}\cap \Forb)$
			of forbidden words can be computed in $\mathcal{O}(1)$ arithmetic operations.
			Furthermore, updating $\PrefixTree$ through the addition of a generated word can be performed in $\Theta(n)$ arithmetic operations.
		\end{prop}
		\begin{proof}
		Thanks to the unambiguity of the grammar, a mature word $v\in\mathcal{L}$ belongs to the
		language generated from an immature one $\omega\in\Imm$ iff $\omega$ is found on the parse
		walk of $v$. Therefore we gather the parse walks associated with forbidden words $\Forb$
		into a \Def{prefix tree}, additionally storing at each node the total weight of all words
		accessible from it. It is then possible to traverse the prefix tree during the generation,
		read the values of $k_\Pond$ for candidate derivations on the children of the current
		node.

		The update of the prefix tree after a generation can be performed in $\Theta(n)$ arithmetic
		operations in two phases, as illustrated by Figure~\ref{fig:PrefixTree}: First, a \emph{top-down} stage adds nodes
		for each immature word traversed during the generation ({\bf a}); Then,
		a \emph{bottom-up} stage will propagate the weight of the sampled mature word to his ancestors ({\bf b}).
		\end{proof}

		Finally, it is remarkable that many of the internal nodes in the trees have degree $1$, which
		means that their value for $k_\Pond$ is just a copy of their child's own. Since in any tree the
		number of internal nodes of degree at least two is strictly smaller than the total number of
		leaves, then the number of \emph{different} values for $k_\Pond$ is bounded by $2*|\Forb|$.
		It follows that the memory needed to store the prefix tree will scale like $\mathcal{O}(n|\Forb|)$,
		even if the encoding of each $k_\Pond$ requires $\Theta(n)$ bits of storage.
	\subsection{Miscellaneous}
		For point~\ref{alg:Boustrophedon}, we can use the so-called Boustrophedon strategy, which allows
		for an $\mathcal{O}(n\log(n))$ arithmetic operations generation in the worst case scenario. Since
		we only restrict the generation set to authorized (or not previously generated) words, such a property
		should hold in our case.

		For point~\ref{alg:Floats}, it is reasonable, for all practical purpose, to assume that the weights are going to be
		expressed as rational numbers. Multiplying these weights by the
		least common multiple of their denominators yields a new set of integral weights inducing the same probability
		distribution, thus arbitrary precision integers can be used. The numbers will scale like $\mathcal{O}(\alpha^n)$
		for some explicit $\alpha$, since the resulting language is context-free, and operations performed on such numbers will
		take time $\mathcal{O}(n\log(n)\log\log(n))$~\cite{Van02}, while the space occupied by their encoding is in $\mathcal{O}(n)$.

	\subsection{Summary}
			Let $n\in\mathbb{N}^+$ be the overall length for generated words, $k\in\mathbb{N}^+$ the number
			of distinct generated words and $\Forb$ the initial set of forbidden parse walks:
			\begin{itemize}
			\item The {\bf time-complexity} of Algorithm $\AlgW$ is in $\Theta(kn\log(n))$ {\bf arithmetic operations} in the worst case
			scenario, after a pre-processing in $\Theta(|\Forb|n + n)$ {\bf arithmetic operations}.
			\item The {\bf memory complexity} is in $\Theta(n)$ {\bf numbers} for the pre-processing, plus $\Theta((|\Forb|+k)n)$ {\bf bits} for the storage
			of the prefix tree.
			\item For rational-valued weights, using arbitrary arithmetics, the associated bit-complexities are in respectively
			$\Theta(kn^2\log(n))$ for time and $\Theta((|\Forb|+k+n)n)$ for memory.
			\item Lastly, starting from an empty forbidden set $\Forb=\emptyset$ yields a generation in $\Theta(kn^2\log(n))$ for time and
			$\Theta(kn+n^2)$ for memory, complexities similar to that of the \Def{possibly redundant} traditional
			approach~\cite{deniserandom,flajoletcalculus}.
			\end{itemize}

\section{Conclusion and perspectives}
	We addressed the random generation of non-redundant sets of sequences from context-free languages,
	while avoiding a predefined set of words. We first investigated the efficiency of a rejection
	approach. Such an approach was found to be acceptable in the uniform case. By contrast, for weighted
	languages, we showed that for some languages the expected number of rejections would grow exponentially
	on the desired number of generated sequences. Furthermore, we showed that in typical context-free languages and
	for fixed length, the probability distribution can be dominated by a small number of sequences.
	We proposed an alternative algorithm solution for this problem, based on the so-called recursive approach. The correctness of the
	algorithm was demonstrated, and its efficient implementation discussed. This algorithm was showed to achieve
	the generation of a non-redundant set of $k$
	structures with a time-complexity in $\mathcal{O}(kn^2\log(n))$, while using $\mathcal{O}(kn+n^2)$ bits of storage.
	These complexities hold in the worst-case scenario, are almost unaffected by to the weights function used, and are equivalent to that
	of the traditional, possibly redundant, generation of $k$ words using previous approaches.

	One natural extension of the current work concerns the random generation of the more general class of decomposable
	structures~\cite{flajoletcalculus}. Indeed, such aspects like the \emph{pointing} and \emph{unpointing} operator
	are not explicitly accounted for in the current work. Furthermore, the generation of labeled structures might
	be amenable to similar techniques in order to avoid a redundant generation. It is unclear however how
	to extend the notion of parse tree in this context. Isomorphism issues might arise, for instance while using
	the \emph{unranking} operator.

	Following the remark that the random generation from reasonable specifications is a \emph{numerically
	stable problem}~\cite{DeZi99}, we could envision using arbitrary precision arithmetics to achieve a
	$\mathcal{O}(kn^{1+\varepsilon})$ complexity. Such a feature could accelerate an implementation of this
	algorithm, for instance in the software {\tt GenRGenS}~\cite{PoTeDe06} that already supports
	the formalism of weighted grammars. Another direction for an efficient implementation of this approach
	would be to investigate the use of Boltzmann samplers~\cite{fullboltz}.

	Moreover, the influence of the number of desired sequences, the length and the weights over the complexity
	of a rejection based approach deserves to be further characterized. Namely, are there \emph{simple-type} grammars
	giving rise to an exponential complexity on $k$ ? Can \emph{phase transition}-like phenomena be observed for varying
	weights ?

\section*{Acknowledgements}
	The author would like to thank A. Denise, for his many helpful comments and suggestions,
	and C. Herrbach, for asserting that some parts of the document actually made some sense.
\bibliographystyle{amsplain}
\bibliography{biblio}
\end{document}